\documentclass[runningheads]{llncs}
\usepackage{graphicx}
\usepackage[T1]{fontenc}
\usepackage[utf8]{inputenc}
\usepackage{thmtools, thm-restate}

\usepackage{amsmath}
\usepackage{amssymb}

\usepackage{calc}
\usepackage{etoolbox}

\usepackage{tikz}
\usetikzlibrary{graphs}
\usetikzlibrary{graphs.standard}

\spnewtheorem{myclaim}{Claim}{\itshape}{\rmfamily}

\newcommand{\conv}{\mathrm{conv}}
\newcommand{\ie}{i.\,e.}
\newcommand{\etal}{et al.\ }

\DeclareMathOperator{\intSol}{Sol.\!int}
\DeclareMathOperator{\fracSol}{Sol.\!frac}
\DeclareMathOperator{\prox}{prox}
\DeclareMathOperator{\sens}{sens}
\DeclareMathOperator{\dist}{dist}
\DeclareMathOperator{\subDet}{subDet}

\newcommand{\svdots}{\raisebox{3pt}{$\scalebox{.6}{$\vdots$}$}}
\newcommand{\sdots}{\raisebox{3pt}{$\scalebox{.6}{$\dots$}$}}
\newcommand{\sddots}{\raisebox{3pt}{$\scalebox{.6}{$\ddots$}$}}

\begin{document}

\title{Tightness of Sensitivity and Proximity Bounds for Integer Linear Programs\thanks{This work was supported by DFG project JA 612/20-1}}

\authorrunning{S. Berndt et al.}

\author{Sebastian Berndt\inst{1}\orcidID{0000-0003-4177-8081} \and
Klaus Jansen\inst{2} \and
Alexandra Lassota \inst{3}\orcidID{0000-0001-6215-066X}}

\institute{Institute of IT Security, University of Lübeck,  Lübeck, Germany
\email{s.berndt@uni-luebeck.de}\\ \and
Department of Computer Science, Kiel University,  Kiel, Germany\\
\email{kj@informatik.uni-kiel.de}\\ \and
Department of Computer Science, Kiel University,  Kiel, Germany\\
\email{ala@informatik.uni-kiel.de}}

\maketitle

\begin{abstract}
  We consider Integer Linear Programs (ILPs), where each variable corresponds to an
  integral point within a polytope $\mathcal{P}\subseteq \mathbb{R}^{d}$, i.\,e., ILPs of the form
  $\min\{c^{\top}x\mid \sum_{p\in\mathcal P\cap \mathbb Z^d} x_p p = b,
  x\in\mathbb Z^{|\mathcal P\cap \mathbb Z^d|}_{\ge 0}\}$.
  The distance between an optimal fractional solution and an optimal
  integral solution (called the \emph{proximity}) is an important measure. A classical result by Cook et
  al.~(Math. Program., 1986) shows that it is at most $\Delta^{\Theta(d)}$ where
  $\Delta=\lVert \mathcal{P}\cap \mathbb{Z}^{d} \rVert_{\infty}$ is the largest
  coefficient in the constraint matrix.
  Another important measure studies the change in an optimal solution if the
  right-hand side $b$ is replaced by another right-hand side $b'$. The 
  distance between an optimal solution $x$ w.r.t.~$b$ and an optimal solution
  $x'$ w.r.t.~$b'$ (called the \emph{sensitivity}) is similarly bounded, i.\,e., $\lVert
  b-b' \rVert_{1}\cdot \Delta^{\Theta(d)}$, also shown by Cook \etal~(Math. Program., 1986). 

  Even after more than thirty years, these bounds are essentially the best known
  bounds for these measures.
  While some lower bounds are known for these measures, they either only work
  for very small values of $\Delta$, require negative entries in the
  constraint matrix, or have fractional right-hand sides. 
  Hence, these lower bounds often do not correspond to instances from
  algorithmic problems. 
 This work presents for each $\Delta > 0$ and each $d > 0$ ILPs of the
  above type with non-negative constraint matrices such that their proximity and
  sensitivity is at least $\Delta^{\Theta(d)}$.
  Furthermore, these instances are closely related to instances of the Bin
  Packing problem as they form a subset of columns of the \emph{configuration ILP}.
  We thereby show that the results of Cook et al.~are indeed tight, even for
  instances arising naturally from problems in combinatorial optimization.

\end{abstract}
  
\section{Introduction}
Integer (Linear) Programs are of great interest throughout computer science, both in
theory and in practice.
Many natural parameters were studied to describe the properties of such
programs.
Let $d \in \mathbb{N}_{>0}$.
For a point $x\in \mathbb{R}^{d}$ and a set $Y\subseteq \mathbb{R}^{d}$, we
define $\dist(x,Y)$ as the minimal $\ell_{\infty}$-distance of $x$ to any point
in $Y$, i.\,e., $\dist(x,Y)=\min_{y\in Y}\{\lVert x-y \rVert_{\infty}\}$.
Furthermore, for two sets $X,Y\in \mathbb{R}^{d}$, we define
$\dist(X,Y)=\max_{x\in X}\{\dist(x,Y)\}$ as the maximum over all minimal distances between any point $x\in X$ to the
set $Y$. 
This work focuses on two such measures called \emph{sensitivity} and
\emph{proximity} that frequently arise in the design of approximation and online algorithms (see
e.\,g.~\cite{DBLP:journals/mp/EpsteinL09,DBLP:journals/siamjo/EpsteinL13,DBLP:journals/orl/Hochbaum04,DBLP:journals/siamdm/JansenK19,DBLP:conf/waoa/JansenR11,DBLP:journals/mor/SandersSS09,DBLP:journals/mor/SkutellaV16,DBLP:journals/mlq/Subramani04}). 
For a given constraint matrix $A\in \mathbb{Z}^{d\times n}$, a right-hand side
$b\in \mathbb{Z}^{d}$, and an objective
function $c\in \mathbb{Z}^{n}$, let $\intSol(A,b,c) = \{x \mid Ax = b, \min\{c^{\top}x\}, x\in
\mathbb{Z}^{n}_{\geq 0}\}$ be the set of optimal integral solutions. Further, denote by $\fracSol(A,b,c) = \{z | Az = b, \min\{c^{\top}z\}, z\in
\mathbb{Q}^{n}_{\geq 0}\}$ the set of optimal fractional solutions, i.e., the integrality constraint $x \in \mathbb{Z}^{n}_{\geq 0}$ is relaxed to $z \in \mathbb{Q}^{n}_{\geq 0}$.
Throughout this work, we always assume that an ILP has $n$ variables, $d$
constraints, and is of full rank. Thus $n\geq d$ holds.

The \emph{sensitivity} of the ILP measures the distance
between two optimal integral solutions if the right-hand side changes.
Formally, we define $\sens(A,b,b',c)$ as $\dist(\intSol(A,b,c), \intSol(A,b',c))$. 
A small sensitivity is useful when the right-hand side changes in a problem
formulation as this implies that an optimal solution for the new problem is
close. Thus, we do not have to change our current
optimal integral solution $x$ too much. Hence, we can just search for it
exhaustively or by a dynamic program. Typical applications are online algorithms where new items arrive or
leave (thus changing the right-hand side corresponding to the present items).

The \emph{proximity} of the ILP denoted by $\prox(A,b,c)$ is formally
defined as the term 
$\dist(\fracSol(A,b,c), \intSol(A,b,c))$, i.\,e., the maximal distance between any optimal fractional
solution and an optimal integral one. 
If the proximity is small, \ie, there exists an optimal integer solution near to any optimal
fractional solution, this allows us to solve the Integer Linear Program fast: First, we
compute the optimal fractional solution $z$, then we search for an optimal
integral solution $x$ in the small box implied by the proximity bound around~$z$.

Cook \etal presented in \cite{DBLP:journals/mp/CookGST86} upper bounds for these
values. In the following, $\Delta$ will always denote the largest absolute value of
the entries in $A$, i.\,e., $\Delta = \lVert A \rVert_{\infty}$ and
$\subDet(A)$ will be the largest determinant of any $d\times d$ submatrix of~$A$. Note that this value bounds the determinant of any submatrix of any dimension as $A$ is of full rank. 
\begin{proposition}[Theorem 1 in~\cite{DBLP:journals/mp/CookGST86}]\label{p:CookProximity}
  If $\intSol(A,b,c)$ is non-empty, then for each $x\in \intSol(A,b,c)$ we have
  $\dist(x,\fracSol(A,b,c))\leq n \cdot \subDet(A)$ and furthermore, for each $y\in
  \fracSol(A,b,c)$ we have $\dist(y,\intSol(A,b,c))\leq n\cdot \subDet(A)$. 
\end{proposition}
Note that this implies that $\prox(A,b,c)\leq n\cdot \subDet(A)$. 

\begin{proposition}[Theorem 5 in~\cite{DBLP:journals/mp/CookGST86}]\label{p:CookSensitivity}
  If both $\intSol(A,b,c)$ and $\intSol(A,b',c)$ are non-empty, we have $\dist(x,\intSol(A,b',c))\leq (\lVert b-b'
  \rVert_{\infty}+2)\cdot n\cdot \subDet(A)$  for each
  $x\in \intSol(A,b,c)$.
\end{proposition}
Note that this implies that $\sens(A,b,b',c)\leq (\lVert b-b'
\rVert_{\infty}+2)\cdot n\cdot \subDet(A)$. 

The Hadamard inequality states that the determinant of a quadratic matrix $N^{n \times n}$ with with columns $n_i$ is bounded by $\text{det}(N) \leq \prod_{i=1}^n \lVert n_i \rVert$ \cite{hadamard1893resolution}. This implies that $\subDet(A) \leq \Delta^{d}\cdot
d^{d/2}$. As $n \leq (2 \Delta + 1)^d$ (maximum number of distinct
columns), we can bound $n \cdot \subDet(A)$ by $((2 \Delta + 1)^d) \cdot d =
\Delta^{\Theta(d)}$.
Surprisingly, these bounds do not depend on the objective function $c$ nor on the
size of $b$ (only the sensitivity depends on the distance between $b$ and $b'$)
but only on the matrix~$A$. We will thus often drop the objective function from
our notation and write $\prox(A,b)$ (resp.~$\sens(A,b,b')$) to reflect this. 

While it is known that these bounds are tight, all known examples either have a
very small value of $\Delta=1$,  use negative entries in the constraint
matrix, and have a non-integral right-hand side~\cite{DBLP:books/daglib/0090562}. 

Hence, these lower bounds often do not correspond to instances from
algorithmic problems.
Nevertheless, knowing the exact bounds is often helpful.
For example the exponent denoted $C(A_{\delta})$ in the running time of the algorithm
in~\cite{DBLP:conf/waoa/JansenR11} is just an upper bound on the proximity of
the underlying configuration IP. Hence, improving this upper bound would
directly lead to a better running time.
Another example concerning the sensitivity comes from the field of online
algorithms.
Often times, the requirement that decisions are not allowed to be rewinded is
too strict.
Hence, \cite{DBLP:journals/mor/SandersSS09} introduced the model of the
\emph{migration factor} where a bounded amount of rewinding is allowed.
The migration factor in their work and in many others
(e.\,g.~\cite{DBLP:journals/mp/EpsteinL09,DBLP:journals/siamjo/EpsteinL13,DBLP:journals/siamdm/JansenK19,DBLP:journals/mor/SkutellaV16})
are simply given by the sensitivity of the underlying IPs. 
Again, any improvement on the general sensitivity results would directly improve
these migration factors. 

This work presents for each $\Delta > 0$ and each $d > 0$ ILPs of the
above type with non-negative constraint matrices such that their proximitiy and
sensitivity are at least $\Delta^{\Theta(d)}$. Note that $\Delta > 0$ and $d >
0$ can be chosen arbitrarily large, however,  we restrict $d$ to be odd or even depending on the case. 

\begin{restatable}{theorem}{sensitivity}
  \label{t:Sensitivity}
  For each $\Delta > 0$ and each even $d > 0$, there is a non-negative matrix
  $A \in \mathbb{Z}_{\geq 0}^{d\times d}$, a right-hand side
  $b \in \mathbb{Z}_{\geq 0}^{d}$, and a right-hand side
  $b' \in \mathbb{Z}_{\geq 0}^{d}$ with $\lVert b-b'
  \rVert_{1}=1$ such that
  $\sens(A,b,b')\geq \Delta^{\Theta(d)}$. Furthermore, the underlying ILP is
  polytopish. 
\end{restatable}

\begin{restatable}{theorem}{proximity}
  \label{t:Proximity}
  For each $\Delta \geq 2$ and each odd $d > 0$, there is a non-negative matrix
  $A \in \mathbb{Z}_{\geq 0}^{15d\times 15d+6}$ and a right-hand
  side $b \in \mathbb{Z}_{\geq 0}^{15d}$ such that
  $\prox(A,b)\geq \Delta^{\Theta(d)}$.
  Furthermore, the underlying ILP is
  polytopish. 
\end{restatable}

\paragraph*{Polytopish Integer (Linear) Programs}
This work considers a special case of integer (linear) programs where each
variable corresponds to an integral point within a polytope $\mathcal P\subseteq \mathbb R^d$. The corresponding Integer Linear Program is defined by 
\begin{align*}
 &  \min c^{\top}x \\
 &  \sum_{p\in\mathcal P\cap \mathbb Z^d} x_p p = b \\
 &  x \in\mathbb Z^{|\mathcal P \cup \mathbb{Z}^d|}_{\ge 0} .
\end{align*}
We call such an ILP \emph{polytopish}. 

  These ILPs often arise in the context of algorithmic applications.
  Probably the most famous one among such ILPs is the \emph{configuration ILP} introduced by
  Gilmore and Gomory~\cite{gilmore1961linear} and used for many
  packing and scheduling problems (e.\,g.~\cite{alon1998approximation,DBLP:conf/soda/GoemansR14,DBLP:conf/innovations/JansenKMR19,DBLP:conf/icalp/JansenKV16,DBLP:journals/corr/abs-1909-11970}).
  
The origin of the configuration ILP lies in the Bin Packing problem. There we are given $n$ items with sizes $s_1, \dots, s_n\leq 1$. The objective is to pack these items into as few unit-sized bins as possible. 
As some
sizes may be equal, i.\,e., $\{s_{1},\ldots,s_{n}\}=\{s_{1},\ldots,s_{d}\}$ for
some $d < n$, we can rewrite the instance as a multiplicity vector of sizes,
\ie, $(b_1, \dots, b_d)$ where the $i$th item size occurs $b_i$ times. Further
we define configurations. A configuration $k = (k_1, \dots, k_d) \in
\mathbb{Z}^d_{\geq 0}$ is a multiplicity vector of item sizes such that the sum
of their sizes is at most $1$, \ie, $k \cdot (s_1, \dots, s_d)^T \leq
1$, hence a feasible packing for a bin. Define the constraint matrix as the set
of feasible configurations (one configuration per column). We now aim to find a
set of configurations such that we cover each item, \ie, the multiplicities of
the item sizes of the chosen configurations equal the occurrences of the item
sizes from the input.
The goal is to minimize the number of used configurations (including how often they are chosen).
Note that all fractional solutions to the constraint $k\cdot
(s_{1},\ldots,s_{d})^{\top}\leq 1$ describe a knapsack polytope $\mathcal{P} := \mathcal{P}_{s_{1},\ldots,s_{d}}$.
Hence, the integer linear program (\ie, the \emph{configuration ILP}) can be written as

\begin{align*}
 &  \min \lVert x \rVert_{1} \\
 & \sum_{p\in\mathcal{P}\cap \mathbb Z^d} x_p p = b \\
 & x \in\mathbb Z^{|\mathcal P \cup \mathbb{Z}^d|}_{\ge 0} .
\end{align*}

We define item sizes such that the set $C_1$ of columns in the examples for the general ILPs are a subset of configurations for the Bin Packing problem. Then we
define an objective function where all values corresponding to the
configurations $k \in C_1$ get value $0$ and the remaining ones value $1$. To
minimize this function we thus cannot take other columns than the ones in~$C_1$.
Setting the right-hand side as for the general ILPs this essentially yields the
same examples. Thus the same bounds are achieved.
This construction shows that in order to improve bounds on the proximity or
sensitivity of the Bin Packing problem the objective function $\min \lVert x
\rVert_{1}$ needs to be taken into account.

\paragraph*{Related Work}
Already in 1986, Cook \etal proved upper bounds regarding the proximity and
sensitivity for general Integer Linear
Programs~\cite{DBLP:journals/mp/CookGST86}, see
Proposition~\ref{p:CookSensitivity} and Proposition~\ref{p:CookProximity}.
Still, these classical bounds are state-of-the-art. This rises the question if these bounds are tight. In this work we answer this affirmatively.

For the case of $d = 1$, Aliev \etal present a tight lower bound regarding
proximity of $\lVert x-z\rVert_{\infty} \leq \Delta -
1$~\cite{aliev2019distances}. Further settings were studied, such as separable convex objective functions~\cite{DBLP:journals/jacm/HochbaumS90} or mixed integer constraints~\cite{DBLP:journals/mp/PaatWW20}.

Recently, another proximity bound independent of $n$ was proven by Eisenbrand and Weismantel~\cite{DBLP:journals/talg/EisenbrandW20}. Using the Steinitz lemma, they show that the $\ell_1$-distance of an optimal fractional solution $z$ and its corresponding integral solution $x$ is bounded by $\lVert x-z\rVert_1 \leq m \cdot (2m\Delta+1)^m$. This result also holds when upper bounds for the variables are present. This result is improved to $\lVert x-z\rVert_1 < 3m^2 \log(2\sqrt{m} \cdot \Delta^{1/m})\cdot \Delta$ using sparsity~\cite{lee2020improving}. 

For a special sub-case of Integer Linear Programs where the constraint matrix
consists of non-zero entries only in the first $r$ rows and in blocks of size $s
\times t$ in the diagonal beneath, sensitivity and proximity results were also obtained. For these so-called $n$-fold ILPs it holds that if $x$ is a solution to a right-hand side $b$ and the right-hand side changes to $b'$ still admitting a finite, optimal solution $x'$ then $\lVert x-x' \rVert_1 \leq \lVert b-b'\rVert_1 \cdot O(rs\Delta)^{rs}$~\cite{DBLP:conf/icalp/JansenLR19}.  In turn, it was shown that the proximity is bounded by $\lVert x-z\rVert_1 \leq (rs\Delta)^{O(rs)}$~\cite{DBLP:journals/corr/abs-2002-07745}. Note that both bounds are independent of the number of rows and columns of the complete constraint matrix.
  
\section{Sensitivity of ILPs} 
This section provides lower bounds for the sensitivity of ILPs. First, we present an example for general ILPs. Then we show how we can use this example to prove the same bound for the ILP which arises from the Bin Packing polytope. 

\subsection{Sensitivity of General ILPs}
This section proves the sensitivity bound for general ILPs, i.\,e., a lower bound
on the  distance
between $\intSol(A,b)$ and $\intSol(A,b')$. 
Let $d$ be an even number and $\Delta\in\mathbb{N}_{>0}$.
We consider the following ILP (I) with an objective function $c\equiv
\textbf{0}$ (corresponding to no objective function).

  \begin{equation*}
  \tag{I}
    \underbrace{\begin{pmatrix}	
      1 		& 0			&  \sdots  	& 0 		& 0 \\
      \Delta & 1  		& \sdots 	& 0 		& 0 \\
      0 		& \Delta 	&  \sdots 	& 0 		& 0 \\
      \svdots & \svdots 	&  \sddots 	& \svdots & \svdots \\
      0 		& 0 		&  \sdots 	& 1 		& 0 \\
      0 		& 0 		& \sdots 	& \Delta & 1 \\
    \end{pmatrix}}_{=: A}
    x
    =
    \begin{pmatrix}
     1 \\
     \Delta \\
     \Delta^2 \\
     \svdots \\
     \Delta^{d-2} \\
     \Delta^{d-1}
    \end{pmatrix}
  \end{equation*}

  Note that ILP (I) is polytopish.
  To see this, let $A_{1},\ldots,A_{d}$ be the columns of the ILP and define
  $\mathcal{P}=\conv\{A_{1},\ldots,A_{d}\}$ as the convex hull of the columns. 
In the following, we prove that the integer points in $\mathcal{P}$ are exactly the columns themselves.

 \begin{myclaim}
   \label{claim:sensitivity:polytopish}
It holds that $\conv \{A_1, A_2,\dotsc, A_d\} \cap \mathbb Z^{d} = \{A_1, A_2, \dotsc, A_d\}$.
\end{myclaim}
\begin{proof}
Let $x_1, \dotsc, x_n$ be a convex combination of the columns where $Ax$ is an integer point.
Let $i$ be the first column with $0 < x_i < 1$. The $i$th row appears with a non-zero entry
only in the columns $i$ and $i-1$. Since $x_{i-1}$ is not fractional, the $i$th entry
of $Ax$ is fractional. This is a contradiction. Hence, there are no fractional values in $x$
and therefore exactly one is $1$ and all others are $0$. \qed
\end{proof}
Hence, the ILP is polytopish. 
Next we prove the main result of this section.

\sensitivity*
\begin{proof}
An optimal solution to the ILP above is clearly unique (Note that we set the objective function to zero, thus optimality corresponds to feasibility.). We have only one column with a non-zero entry for the first row. Thus, the high-hand side $b$ determines this value. By that, we have only one free, non-zero variable for the second row. Using this argument inductively we get a unique solution of form $x = (1, 0, \Delta^2, 0, \Delta^4, \dots, \Delta^{d-2}, 0)$. 

If we now change the first entry of the right-hand side to $0$, we get again a
unique solution for $b'$ due to the same argument as above. The solution is of
form: $x' = (0, \Delta, 0, \Delta^3, \dots, \Delta^{d-1})$. Obviously, the
difference is $\lVert x - x'\rVert_1 \geq \lVert b-b'\rVert_1
\Delta^{\Theta(d)}$ implying the statement.
The ILP is polytopish due to Claim~\ref{claim:sensitivity:polytopish}. 

\qed
\end{proof}

\subsection{Sensitivity of the Bin Packing ILP}
Let us now construct an example where the sensitivity for the Bin Packing
polytope is large. In this problem, we are given $n$ items with $d$ different
sizes. Define these sizes as $s_i = 1/(2\Delta) + i \cdot \epsilon$ for $i = 1,
\dots, d$ and some $\epsilon > 0$ with $\epsilon \leq \frac{1}{4(d-1+\Delta
  d)}$. Obviously, the constraint matrix from the previous example is a subset
of feasible configurations, \ie, a subset of the columns of the constraint
matrix for this problem, as
\begin{align*}
  &s_{i}+\Delta s_{i+1}\leq s_{d-1}+\Delta s_{d}=1/(2\Delta)+(d-1)\epsilon + 1/2 +\Delta d \epsilon =\\
  &1/2 + 1/(2\Delta)+ \epsilon (d-1+\Delta d)\underbrace{\leq}_{d \geq 2} 1/2 + 1/4 +\epsilon (d-1+\Delta d) \\
  & \underbrace{\leq}_{\epsilon \leq \frac{1}{4(d-1+\Delta
      d)}} 1/2 + 1/4 + 1/4 = 1.
\end{align*}
Define by $C_1$ the set of these columns. Let us now define a linear objective function $c$ which has a $0$ entry for each configuration $k \in C_1$ and $1$ otherwise. Thus to minimize the objective function we can only choose configurations from~$C_1$. Setting and changing the right-hand side as in the previous example will clearly lead to the same sensitivity bound. Combining it with the result of Cook \etal we thus get:

\begin{corollary}
  There is an objective function $c$ such that for the configuration ILP with
  constraints $A$ and right-hand sides $b$ and $b'$ we have
  $\sens(A,b,b',c)\geq \Delta^{\Theta(d)}$. 
\end{corollary}
Hence, if one aims to improve the sensitivity of the configuration ILP,
the special objective function $\lVert x \rVert_{1}$ needs to be taken into
account.

\section{Proximity of ILPs}
This section presents an example for general ILPs, where the optimal integer
solution $x$ differs greatly from the corresponding fractional solution $z$, \ie,  $\lVert x - z\rVert_1 =\Delta^{\Theta(d)}$. By this, we give a lower bound on the proximity of general ILPs which meets the upper bound for ILPs shown by Cook \etal implying their tightness. Further we use this example to construct an instance of the Bin Packing problem where the same bound is met.

\subsection{Proximity of General ILPs}
To construct this example we make use of the Petersen graph. This graph $P = (V,
E)$ has fifteen edges, ten vertices and six perfect matchings. A perfect matching $M$ is a set of edges such that each vertex $v \in V$ is part of exactly one edge, i.\,e., there exists exactly one edge $e = (u, w) \in M$ satisfying $v = u$ or $v = w$. 
The Petersen graph has the nice property that every edge is part of exactly two perfect matchings and every two perfect matchings share exactly one edge \cite{akiyama2011matchings}. The graph and its perfect matchings are displayed in Figure~\ref{f:Petersen}.
  
    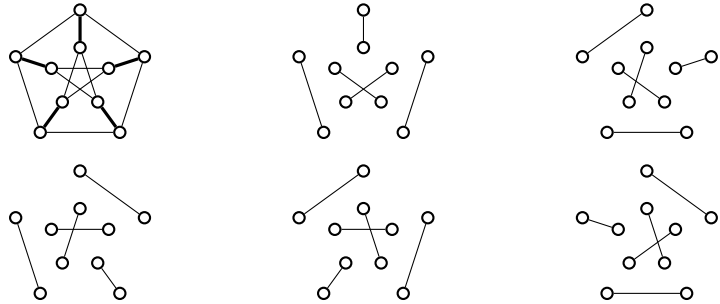
\begin{figure}[h]
    \centering
    
\begin{minipage}[b]{0.3\linewidth}
    \centering
   \begin{tikzpicture}[every node/.style={draw,circle,inner sep  = 1.5, thick}]
  \foreach \name/\angle/\text in {a-1/234/, a-2/162/, 
                                  a-3/90/, a-4/18/, a-5/-54/}
    \node[xshift=9cm,yshift=.5cm] (\name) at (\angle:0.9cm) {$\text$};
    
      \foreach \name/\angle/\text in {b-1/234/, b-2/162/, 
                                  b-3/90/, b-4/18/, b-5/-54/}
    \node[xshift=9cm,yshift=.5cm] (\name) at (\angle:0.4cm) {$\text$};

\draw (a-1) -- (a-2);
\draw (a-2) -- (a-3);
\draw (a-3) -- (a-4);
\draw (a-4) -- (a-5);
\draw (a-5) -- (a-1);

\draw (b-1) -- (b-3);
\draw (b-1) -- (b-4);
\draw (b-2) -- (b-4);
\draw (b-2) -- (b-5);
\draw (b-3) -- (b-5);

\draw[very thick] (a-1) -- (b-1);
\draw[very thick] (a-2) -- (b-2);
\draw[very thick] (a-3) -- (b-3);
\draw[very thick] (a-4) -- (b-4);
\draw[very thick] (a-5) -- (b-5);
\end{tikzpicture}
\end{minipage}
\begin{minipage}[b]{0.3\linewidth}
    \centering
   \begin{tikzpicture}[every node/.style={draw,circle, inner sep = 1.5, thick}]
   \foreach \name/\angle/\text in {a-1/234/, a-2/162/, 
                                  a-3/90/, a-4/18/, a-5/-54/}
    \node[xshift=9cm,yshift=.5cm] (\name) at (\angle:0.9cm) {$\text$};
    
      \foreach \name/\angle/\text in {b-1/234/, b-2/162/, 
                                  b-3/90/, b-4/18/, b-5/-54/}
    \node[xshift=9cm,yshift=.5cm] (\name) at (\angle:0.4cm) {$\text$};

\draw (a-4) -- (a-5);
\draw (a-1) -- (a-2);
\draw (a-3) -- (b-3);
\draw (b-1) -- (b-4);
\draw (b-2) -- (b-5);
\end{tikzpicture}
\end{minipage}
\begin{minipage}[b]{0.3\linewidth}
    \centering
   \begin{tikzpicture}[every node/.style={draw,circle, inner sep = 1.5, thick}]
   \foreach \name/\angle/\text in {a-1/234/, a-2/162/, 
                                  a-3/90/, a-4/18/, a-5/-54/}
    \node[xshift=9cm,yshift=.5cm] (\name) at (\angle:0.9cm) {$\text$};
    
      \foreach \name/\angle/\text in {b-1/234/, b-2/162/, 
                                  b-3/90/, b-4/18/, b-5/-54/}
    \node[xshift=9cm,yshift=.5cm] (\name) at (\angle:0.4cm) {$\text$};

\draw (a-2) -- (a-3);
\draw (a-1) -- (a-5);
\draw (a-4) -- (b-4);
\draw (b-1) -- (b-3);
\draw (b-2) -- (b-5);
\end{tikzpicture}
\end{minipage}

\vspace{0.3cm}

\begin{minipage}[b]{0.3\linewidth}
    \centering
   \begin{tikzpicture}[every node/.style={draw,circle, inner sep = 1.5, thick}]
   \foreach \name/\angle/\text in {a-1/234/, a-2/162/, 
                                  a-3/90/, a-4/18/, a-5/-54/}
    \node[xshift=9cm,yshift=.5cm] (\name) at (\angle:0.9cm) {$\text$};
    
      \foreach \name/\angle/\text in {b-1/234/, b-2/162/, 
                                  b-3/90/, b-4/18/, b-5/-54/}
    \node[xshift=9cm,yshift=.5cm] (\name) at (\angle:0.4cm) {$\text$};

\draw (a-1) -- (a-2);
\draw (a-3) -- (a-4);
\draw (a-5) -- (b-5);
\draw (b-1) -- (b-3);
\draw (b-2) -- (b-4);
\end{tikzpicture}
\end{minipage}
\begin{minipage}[b]{0.3\linewidth}
    \centering
   \begin{tikzpicture}[every node/.style={draw,circle, inner sep = 1.5, thick}]
   \foreach \name/\angle/\text in {a-1/234/, a-2/162/, 
                                  a-3/90/, a-4/18/, a-5/-54/}
    \node[xshift=9cm,yshift=.5cm] (\name) at (\angle:0.9cm) {$\text$};
    
      \foreach \name/\angle/\text in {b-1/234/, b-2/162/, 
                                  b-3/90/, b-4/18/, b-5/-54/}
    \node[xshift=9cm,yshift=.5cm] (\name) at (\angle:0.4cm) {$\text$};

\draw (a-3) -- (a-2);
\draw (a-5) -- (a-4);
\draw (a-1) -- (b-1);
\draw (b-5) -- (b-3);
\draw (b-2) -- (b-4);
\end{tikzpicture}
\end{minipage}
\begin{minipage}[b]{0.3\linewidth}
    \centering
   \begin{tikzpicture}[every node/.style={draw,circle, inner sep = 1.5, thick}]
   \foreach \name/\angle/\text in {a-1/234/, a-2/162/, 
                                  a-3/90/, a-4/18/, a-5/-54/}
    \node[xshift=9cm,yshift=.5cm] (\name) at (\angle:0.9cm) {$\text$};
    
      \foreach \name/\angle/\text in {b-1/234/, b-2/162/, 
                                  b-3/90/, b-4/18/, b-5/-54/}
    \node[xshift=9cm,yshift=.5cm] (\name) at (\angle:0.4cm) {$\text$};

\draw (a-1) -- (a-5);
\draw (a-3) -- (a-4);
\draw (a-2) -- (b-2);
\draw (b-1) -- (b-4);
\draw (b-5) -- (b-3);
\end{tikzpicture}
\end{minipage}
    \caption{The first sub-figure presents the complete Petersen graph with one perfect matching marked by thick edges. The remaining sub-figures each present one of the remaining five perfect matchings.}
    \label{f:Petersen}
  \end{figure}

The Petersen graph is named after its appearance in a paper written by Petersen \cite{Petersen98} in 1898. However, it was first mentioned as far back as 1886 \cite{Kempe86}. This graph is often used to construct counter-examples for various conjectures due to its neat structure and nice properties. For example it was used in \cite{DBLP:journals/mp/CapraraDDIR15} to construct small examples where the Round-up Property for Bin Packing instances does not hold. For a survey concerning this graph and more applications we refer to \cite{holton1993petersen}.

We set up a constraint matrix $M \in \{0, 1\}^{15 \times 6}$ where every row represents an edge in the Petersen graph and every column corresponds to the indicator vector of one of the perfect matchings. Denote by $I$ the identity matrix of size $(15 \times 15)$. An identity matrix is a matrix where all entries are zero except the diagonal being~$1$. Further let $\Delta \ge 2$ and $d$ be an odd number. Construct the ILP~(II) as follows, where the objective function is again zero, i.e.,  $c\equiv
\textbf{0}$:

  \begin{equation*}
  \tag{II}
    \underbrace{\begin{pmatrix}
      M & I & 0 & \sdots & 0 & 0 \\
      0 & \Delta \cdot  I & I & \sdots & 0 & 0 \\
      0 & 0 & \Delta \cdot I & \sdots & 0 & 0 \\
      \svdots & \svdots & \svdots & \sddots & \svdots & \svdots \\
      0 & 0 & 0 & \sdots & I & 0 \\
      0 & 0 & 0 & \sdots & \Delta \cdot I & I \\
    \end{pmatrix}}_{=: A}
    x
    =
    \begin{pmatrix}
    (1, \dots, 1)^T \in \mathbb{N}^{15} \\
    (\Delta, \dots, \Delta)^T \in \mathbb{N}^{15} \\
    ( \Delta^2, \dots,  \Delta^2)^T \in \mathbb{N}^{15} \\
     \svdots \\
     ( \Delta^{d-1}, \dots,  \Delta^{d-1})^T \in \mathbb{N}^{15} \\
     (\Delta^{d}, \dots, \Delta^{d})^T \in \mathbb{N}^{15} \\
    \end{pmatrix}.
  \end{equation*}
  Obviously, the number of columns is $n = 6+15 \cdot d$.
Call the first six columns corresponding to the perfect matchings \emph{matching columns}. 
Further, we want a solution where $z \in [0,1]^{n}$ for the fractional case and $x \in
\{0, 1\}^{n}$ for the integral one.
To show that the ILP is polytopish, we argue as before. First, let
$A_{1},\ldots,A_{n}$ be the columns of $A$ and define
$\mathcal{P}=\conv\{A_{1},\ldots,A_{n}\}$. 
Next, we argue that the integer points in $\mathcal{P}$ are again only the columns themselves.

\begin{myclaim}
  \label{claim:proximity:polytopish}
It holds that $\conv \{A_1, A_2,\dotsc, A_n\} \cap \mathbb Z^{6+15d} = \{A_1, A_2, \dotsc, A_n\}$.
\end{myclaim}
\begin{proof}
  Let $x_1, \dotsc, x_n\in [0,1]$ with $\sum_{i=1}^{n}x_{i} = 1$ such that 
  $A(x_{1}\ldots,x_{n})$ is integral. 
First suppose that $0 < x_i < 1$ for some column $1 \leq i \leq 6$ corresponding
to a matching column.
Then, to obtain an integral point $A(x_{1},\ldots,x_{n})$, we need 
to also choose another set of columns $J\subseteq \{1,\ldots,21\}\setminus
\{i\}$ with $x_{j} > 0$ for all $j\in J$.
For each such $j$, there is a row $r_{j}$, where column $j$ has a value $0$
where columns $i$ has value $1$, as two matchings only share one edge and the
identity matrix only has one non-zero entry in each column.
Hence, one cannot choose the coefficients $x_{j}$ for $j\in J$ such that
$A(x_{1},\ldots,x_{n})$ is integral and $\sum_{i=1}^{n}x_{i} = 1$ holds.
Hence, the coefficients of the matching columns must be integral.

Now, consider the remaining columns.
Suppose there is a column $i> 6$ with $0 < x_{i} < 1$.
If $i\leq 21$, this column corresponds to the first identity matrix.
As no other column has entries in the first $15$ rows, the resulting point
$A(x_{1},\ldots,x_{n})$ cannot be integral.
If $i \leq 36$, the only other columns that have non-zero entries have index
$\leq 21$ and can thus not be fractional.
Using this argument inductively, we see that all solutions for this ILP are integral and thus the assumption holds.  \qed
 \end{proof}
  
Next we want to estimate the $\ell_1$ norm of a (fractional) solution. Define $p = \sum_{i=1}^{(d-1)/2} \Delta^{2i-1}$ and $q = \sum_{i=1}^{(d-1)/2} \Delta^{2i-2} = \Delta \cdot \sum_{i=0}^{(d-1)/2} \Delta^{2i-1}  = \Delta \cdot p$. 

\begin{myclaim} \label{c:Norm}
The $\ell_1$-norm of any (fractional) solution $x$ is at least $\lVert x \rVert_1 = \lVert y \rVert_1 + \lVert w \rVert_1 \geq \lVert y\rVert_1 + (15 - \lVert y\rVert|_1) \cdot  \Delta \cdot p + \lVert y\rVert_1\cdot p$.
\end{myclaim}
\begin{proof}
Consider a (fractional) solution $x = (y, w)$, where $y$ corresponds to the first $6$ columns, \ie, to the matching columns. Likewise, divide $A=(B, C)$ into matching and non-matching columns. The value for the right-hand side given by $y$ covers some part of the first 15 rows, namely $0 \le \lVert B y \rVert_1 \le 15$, as we can choose at most $3$ columns (fractionally) such that edges are not overlapping. If we would choose more, edges would be overlapping and thus the right-hand side would be greater than $1$ and the solution would be infeasible. Further, each column contains exactly $5$ ones as each perfect matching admits exactly $5$ edges. Combining this we get at most $\lVert y\rVert_1 \cdot 5 \leq 3 \cdot 5 = 15$.

Let $i \le 15$ and $a_i := (B \cdot y)_i$, \ie, the right-hand side covered by the matching columns at position $i$. Thus set $x_{i} = 1 - a_i$ to satisfy the remaining right-hand side. Further,
\begin{align*}
w_{i + 1 \cdot 15} = \Delta - \Delta (1 - a_i) = \Delta (1-(1-a_i)) = \Delta \cdot a_i, 
\end{align*}
as these are the only free variable to satisfy the right-hand side, which already has the value $\Delta (1-a_i)$. In turn, this determines the next $15$ variables, \ie,  
\begin{align*}
w_{i + 2 \cdot 15} = \Delta^2 - \Delta^2 (a_i) = \Delta^2(1-a_i).
\end{align*}
Proceeding with setting the only free variables for the next $15$ rows to be satisfied we get $w_{i + 3 \cdot 15} = \Delta^3a_i$- The scheme proceeds like this. Therefore,

\begin{align*}
  \lVert w \rVert_1 = \sum_{i=1}^{15} [\sum_{j=1}^{(d+1)/2} (1 - a_i) \cdot \Delta^{2i-2} + \sum_{j=1}^{(d-1)/2} a_i \cdot \Delta^{2i-1}] \\
  = \sum_{i=1}^{15} (1-a_i) (\sum_{j=1}^{(d+1)/2} \cdot \Delta^{2i-2}) + \sum_{i=1}^{15} a_i (\sum_{j=1}^{(d-1)/2} \cdot \Delta^{2i-1}) \\
= (15-\lVert a\rVert_1) ( \sum_{j=1}^{(d+1)/2} \cdot \Delta^{2i-2}) + \lVert a_i\rVert_1 (\sum_{j=1}^{(d-1)/2} \cdot \Delta^{2i-1})\\
=(15-\lVert a\rVert_1) q + \lVert a_i\rVert_1 p \geq  (15 - \lVert y\rVert_1) \cdot  \Delta \cdot p + \lVert y\rVert_1\cdot p .
\end{align*}

Thus $\lVert x\rVert_1 = \lVert y \rVert_1 + \lVert w \rVert_1 = \lVert
y\rVert_1 + (15 - \lVert y\rVert_1) \cdot  \Delta \cdot p + \lVert
y\rVert_1\cdot p$ completing the proof.\qed

  \end{proof}

\proximity*
\begin{proof}
An optimal fractional solution $z$ is to take each matching column $1/2$ times, \ie, $z = $
\begin{align*}
 (1/2, \dots, 1/2, 0, \dots 0, \Delta, \dots, \Delta, 0, \dots 0, \Delta^2, \dots, \Delta^2, \dots, \Delta^{d-1}, \dots, \Delta^{d-1}, 0, \dots 0).
 \end{align*}
 Note that again optimality corresponds to feasibility, as we have set the objective function to zero.
 It is easy to verify that this solution is feasible as in the first $15$ columns we have two entries in $M$ with value 1. Taking both $1/2$ often and setting all variables of the identity matrix to zero gives the right-hand side $1$. Then again the values for the remaining columns are determined as explained in Theorem~\ref{t:Sensitivity}.

In turn, an optimal integral solution would either avoid all matching columns or take some of them. In the first case this would lead to take all columns of the first identity matrix and again all other values would be determined and thus the solution looks as follows:
\begin{align*}
x = (0, \dots 0, 1, \dots, 1, 0, \dots, 0, \Delta^2, \dots, \Delta^2, 0, \dots 0, \dots, \Delta^d, \dots, \Delta^d). 
\end{align*}

In the second case, at most one matching column is chosen  as two would already give a too large right-hand side in the first $15$ rows (every two matchings share one edge). Then the identity matrix will take the column corresponding to rows which have a zero entry in the chosen matching column. The remaining solution is then determined by the free variables and the right-hand side as explained in Claim~\ref{c:Norm}. 
  
Now let us look at the difference of the optimal fractional solution $z$ and the
optimal integral ones. It is easy to see that when the optimal solution takes no
matching columns, every non-zero component in $z$ is zero in $x$ and vice versa.
As the sum of both solutions is $\Delta^{\Theta(d)}$, their difference is
$\Delta^{\Theta(d)}$. 
For the other case, where a matching column is used, we get that the matching columns differ in all positions leading to difference of all other positions. Thus
\begin{align*}
\lVert x - z \rVert_1 \ge | \lVert x \rVert_1 - \lVert z \rVert_1 | \ge (1 + p + 2 q) - (15 \cdot \Delta \cdot p)\\
= |1+p+2\Delta p - 15\Delta p| = |(1-13\Delta)p+1| \geq 13 \Delta p = \Delta^{\Theta(d)}.
\end{align*}

Thus the difference between any (optimal) fractional solution $z$ and an
(optimal) integral one $x$ is $\lVert x - z\rVert_1 \geq \Delta^{\Theta(d)}$.
The ILP is polytopish due to Claim~\ref{claim:proximity:polytopish} completing the proof.. \qed
\end{proof}

\subsection{Proximity of the Bin Packing ILP}
We construct an example with a huge proximity by relying on the previous construction. Recall that we are given $n$ items with $d$ different sizes in his problem. Define these sizes as $s_i = 1/(30\Delta) + i \cdot \epsilon$ for $i = 1, \dots, d$ and some $\epsilon > 0$ with $\epsilon \leq \frac{57}{60(d-2+\Delta(d-1))}$. Obviously, the constraint matrix from the previous example is a subset
of feasible configurations, \ie, a subset of the columns of the constraint
matrix for this problem, as the value of the largest configuration is bounded by
\begin{align*}
  &s_{i}+\Delta s_{i+1}\leq s_{d-2}+\Delta s_{d-1}=1/(30\Delta)+(d-2)\epsilon + 1/30 +\Delta (d-1) \epsilon =\\
  &1/30 + 1/(30\Delta)+ \epsilon (d-2+\Delta (d-1))\underbrace{\leq}_{d \geq 2} 1/30 + 1/60 +\epsilon (d-2+\Delta(d-1)) \\
 &\underbrace{\leq}_{\epsilon \leq \frac{1}{2(d-2+\Delta(d-1))}} 1/30 + 1/60 + 57/60 = 1.
\end{align*}

Define by $C_1$ the set of these columns. Let us now define a linear objective function $c$, which has a $0$ entry for each configuration $k \in C_1$ and $1$ otherwise. Thus to minimize the objective function we can only choose these configurations. Minimize the constraint matrix accordingly. Computing a fractional optimal solution and an optimal integral one for the right-hand side of the example above will clearly lead to the same proximity. Combining it with the result of Cook \etal we thus get:

\begin{corollary}
  There is an objective function $c$ such that for the configuration ILP with
  constraint matrix $A$ and right-hand side $b$ we have
  $\prox(A,b,c)\geq \Delta^{\Theta(d)}$. 
\end{corollary}
Hence, if one aims to improve the proximity of the configuration ILP,
the special objective function $\lVert x \rVert_{1}$ needs to be taken into
account. 

\paragraph*{Acknowledgments}
The authors want to thank Lars Rohwedder for enjoyable and fruitful discussions at the beginning of this project.
   
\bibliography{ref}
  
\end{document}